\documentclass[conference]{IEEEtran}
\IEEEoverridecommandlockouts

\usepackage{graphicx}
\usepackage{amssymb}
\usepackage{amsmath}
\usepackage{cite}
\usepackage{subfigure}
\usepackage{mathrsfs}
\usepackage[displaymath,mathlines]{lineno}
\usepackage{color}
\usepackage{tabulary}
\usepackage{multirow}
\usepackage{algpseudocode}
\usepackage{algorithm,algpseudocode}
\usepackage{algorithmicx}
\usepackage{pbox}
\usepackage{multicol}
\usepackage{lipsum}
\usepackage{amsthm}
\usepackage{relsize}
\usepackage{lipsum}
\usepackage{epstopdf}

\newcommand{\doublewidetilde}[1]{{%
		\mathpalette\double@widetilde{#1}%
	}}
	
\makeatletter

\def\BState{\State\hskip-\ALG@thistlm}
\makeatother

\newtheorem{theorem}{Theorem}
\newtheorem{lemma}{Lemma}

\newcounter{eqnback}
\newcounter{eqncnt}

\begin{document}

\title{Joint Pilot Sequence Design and Power Control for Max-Min Fairness in Uplink Massive MIMO}

\author{\IEEEauthorblockN{Trinh Van Chien, Emil Bj\"{o}rnson, and Erik G. Larsson}
	\IEEEauthorblockA{Department of Electrical
		Engineering (ISY), Link\"{o}ping University, SE-581 83 Link\"{o}ping, Sweden\\
		\{trinh.van.chien, emil.bjornson, erik.g.larsson\}@liu.se}
	\thanks{This paper was supported by the European Union's Horizon 2020 research and innovation programme under grant agreement No 641985 (5Gwireless). It was also supported by ELLIIT and CENIIT.}
}

\maketitle

\begin{abstract}
This paper optimizes the pilot assignment and pilot transmit powers to mitigate pilot contamination in Massive MIMO (multiple-input multiple-output) systems. While prior works have treated pilot assignment as a combinatorial problem, we achieve a more tractable problem formulation by directly optimizing the pilot sequences. To this end, we compute a lower bound on the uplink (UL) spectral efficiency (SE), for Rayleigh fading channels with maximum ratio (MR) detection and arbitrary pilot sequences. We optimize the max-min SE with respect to the pilot sequences and pilot powers, under power budget constraints. This becomes an NP-hard signomial problem, but we propose an efficient algorithm to obtain a local optimum with polynomial complexity. Numerical results manifest the near optimality of the proposed algorithm and show significant gains over existing suboptimal algorithms.
\end{abstract}

\IEEEpeerreviewmaketitle

\section{Introduction}
Massive MIMO has recently emerged as a key technology for 5G communications \cite{Marzetta2010a,Larsson2014a,Andrews2014b}, since it can bring significant improvements to the spectral and energy efficiency of cellular networks \cite{Ngo2013a}. By equipping the base stations (BSs) with a large number of antennas, the mutual interference, thermal noise, and small-scale fading can be almost eliminated by virtue of the channel hardening and favorable propagation phenomena \cite{Bjornson2016b,Ngo2014a}. The BS utilizes estimated channel state information (CSI) to achieve these gains, which is generally acquired using UL pilot signals.

To achieve the maximum CSI quality, mutually orthogonal pilot sequences are desirable, but this is impractical since the pilot overhead would be proportional to the total number of users in the entire system. The size of the channel coherence block limits the number of orthogonal pilots, and at most half the block should be used for pilots \cite{Bjornson2016a}.
The consequence is that the pilots need to be reused across cells, which creates so-called pilot contamination \cite{Marzetta2010a,Jose2011b}, where users with the same pilot cause large interference to each other. A pilot reuse factor can be applied to not use the same pilots in neighboring cells, which reduces the pilot contamination at the cost of extra pilot overhead \cite{Yang2015a,Bjornson2016a}. However, some pilot contamination still remains and its impact strongly depends on which users that have the same pilot. The pilot assignment is a combinatorial problem, thus finding the optimal assignment is generally NP-hard \cite{Jin2015a}. This has motivated the design of suboptimal greedy pilot assignment algorithms, which utilize statistical information such as the large-scale fading \cite{Jin2015a, Xu2015a}.
Another way to improve the channel estimation quality is to optimize the transmit powers used for the pilots \cite{ Victor2015b, Guo2015a}.

In this paper, we consider joint optimization of the pilot assignment and pilot powers in Massive MIMO, in contrast to previous works that focused on only one of these components. In particular, we obtain ergodic achievable SE expressions for MR detection and arbitrary pilot sequences. The pilot sequences are treated as optimization variables and we formulate a max-min SE optimization problem, which becomes a signomial program. Due to NP-hardness of signomial programs, we propose a suboptimal approach that finds a local optimum in polynomial time. Numerical results show that this solution is close to the global optimum and can provide great performance improvements over prior works.

\textit{Notations}: The lower bold letters are used for vectors and the upper bold are for matrices. $(\cdot)^T$ and $(\cdot)^H$ stand for the transpose and conjugate transpose, respectively. $\mathbf{I}_{n}$ is the $n \times n$ identity matrix. $\mathbb{E} \{\cdot \}$ denotes  expectation, $\| \cdot \| $ is the Euclidean norm, and $\mathcal{CN} (\cdot, \cdot)$ is the circularly symmetric complex Gaussian distribution.

\section{Pilot Designs for Massive MIMO Systems} \label{Section: System Model}

We consider the UL of a multi-cell Massive MIMO system with $L$ cells.  Each cell consists of a BS equipped with $M$ antennas which serves $K$ single-antenna users. All tuples of cell and user indices  belong to the set 
\begin{equation}
\mathcal{S} = \left\{ (i,t): \; i \in \{ 1, \ldots, L\}, \; t \in \{ 1, \ldots, K \} \right \}.
\end{equation}
The radio channels vary over time and frequency. We divide the time-frequency plane into coherence blocks, each containing $\tau_c$ samples, such that the channel between each user and each BS is static and frequency flat. In each coherence block the users transmit pilot sequences of length $\tau_p$ symbols, while the remaining $\tau_c-\tau_p$ symbols are used for data transmission. In this paper, we focus on the UL, so the fraction $(1 - \tau_p / \tau_c)$ of the coherence block is dedicated to  UL data. We assume $\tau_p \geq 1$ to keep the estimation process feasible and stress that the practical case $\tau_p < KL$ is of key importance since it gives rise to pilot contamination.

\subsection{Proposed Pilot Design} \label{subsection: ProposedPilot}

We aim at optimizing the pilot sequence collection $ \{ \pmb{\psi}_{1,1}, \ldots, \pmb{\psi}_{L,K}\}$, where $\pmb{\psi}_{l,k} \in \mathbb{C}^{\tau_p}$ is the pilot sequence assigned to user~$k$ in cell~$l$. To this end, let us define the $\tau_p$ mutually orthonormal basis vectors $\{\pmb{\phi}_1, \ldots, \pmb{\phi}_{\tau_p} \}$, where $\pmb{\phi}_b \in \mathbb{C}^{\tau_p},$ $\forall b= 1, \ldots, \tau_p$. The corresponding basis matrix is
\begin{equation}
\pmb{\Phi} = [\pmb{\phi}_1, \ldots, \pmb{\phi}_{\tau_p}],
\end{equation}
and it satisfies $\pmb{\Phi}^H \pmb{\Phi} = \mathbf{I}_{\tau_p}$. We assume that each pilot sequence is spanned by these basis vectors. In particular, the pilot sequence of user~$k$ in cell~$l$ is
\begin{equation} \label{eq: ProposedPilotSequence}
\pmb{\psi}_ {l,k}=  \sum_{b =1}^{\tau_p} \sqrt{ \hat{p}_{l,k}^b } \pmb{ \phi}_{b}, \quad \forall l,k,
\end{equation}
where $\hat{p}_{l,k}^b \geq 0$ is the power assigned to the $b$th basis vector.
We stress that the pilot construction in \eqref{eq: ProposedPilotSequence} can create arbitrarily many different orthogonal or non-orthogonal pilots, with arbitrary  total pilot power $\| \pmb{\psi}_{l,k} \|^2 = \sum_{b=1}^{\tau_p} \hat{p}_{l,k}^b$.
We assume that the average pilot power of user~$k$ in cell~$l$ satisfies the power constraint
\begin{equation} \label{eq:Max-Power}
\frac{1}{\tau_p} \sum_{b=1}^{\tau_p} \hat{p}_{l,k}^b \leq P_{\textrm{max},l,k}, \quad \forall  l,k,
\end{equation}
where $P_{\textrm{max},l,k}$ is the maximum pilot power for user~$k$ in cell~$l$. The inner product of two pilot sequences $\pmb{\psi}_ {l,k}$ and $\pmb{\psi}_ {i,t}$ is
\begin{equation} \label{eq: Orthogonal_Property}
\begin{split}
\pmb{\psi}_ {l,k}^H \pmb{\psi}_ {i,t} = \sum_{b =1}^{\tau_p} \sqrt{ \hat{p}_{l,k}^b \hat{p}_{i,t}^b  }.
\end{split}
\end{equation}
 These pilot sequences are orthogonal if every term in the sum is zero, which only happens when the two users allocate their pilot power to disjoint subsets of the basis vectors. Otherwise, the sequences are non-orthogonal and the two users will cause pilot contamination to each other. If the square roots of the powers allocated to the $K$~users in cell~$l$ are gathered in matrix form as
\begin{equation} \label{eq: ProposedPilotPower}
\pmb{P}_l = \begin{bmatrix}
\sqrt{\hat{p}_{l,1}^1}      & \sqrt{\hat{p}_{l,2}^1}     &  \cdots     &  \sqrt{\hat{p}_{l,K}^1} \\
\sqrt{\hat{p}_{l,1}^2}       &  \sqrt{\hat{p}_{l,2}^2}     & \cdots    & \sqrt{\hat{p}_{l,K}^2} \\
\vdots & \vdots &  \ddots      &  \vdots  \\
\sqrt{\hat{p}_{l,1}^{\tau_p}}      &  \sqrt{\hat{p}_{l,2}^{\tau_p}}     &  \cdots   &  \sqrt{\hat{p}_{l,K}^{\tau_p}}
\end{bmatrix} \in \mathbb{R}_{+}^{\tau_p \times K},
\end{equation}
 then the users in cell~$l$ utilize a pilot matrix defined as
\begin{equation} \label{eq: PilotStructure1}
\mathbf{\Psi}_l = [ \pmb{\psi}_{l,1}, \ldots,  \pmb{\psi}_{l,K} ] = \pmb{\Phi} \pmb{P}_l.
\end{equation}
We now describe the difference between this general pilot structure and the prior works, for example \cite{Xu2015a,Zhu2015a, Victor2015b, Guo2015a}. 

\subsection{Other Pilot Designs}
The works \cite{Xu2015a,Zhu2015a} considered the assignment of $\tau_p$ orthogonal pilot sequences using equal pilot power $\hat{p}  \leq \tau_p P_{\max,l,k}$ for every user. Using our notation, the pilot matrix in cell $l$ is
\begin{equation} \label{eq: fixedPilotPower}
\widehat{\pmb{\Psi}}_l = [\hat{\pmb{\psi}}_{l,1}, \ldots, \hat{\pmb{\psi}}_{l,K} ] = \sqrt{\hat{p}}  \pmb{\Phi} \pmb{\Pi}_l,
\end{equation}
where $\pmb{\Pi}_l \in \mathbb{R}_{+}^{\tau_p \times K}$ is a permutation matrix. This matrix is optimized in \cite{Xu2015a,Zhu2015a} to assign the pilots to users to minimize a metric of mutual interference.
This pilot design is a special case of our proposed design, since \eqref{eq: fixedPilotPower} assumes the use of orthogonal pilot sequences and equal power allocation. These assumptions might be suboptimal in systems with large pathloss differences. The selection of the optimal permutation matrix for cell $l$ is a complicated combinatorial problem, so \cite{Xu2015a,Zhu2015a} only study the special case of $\tau_p = K$.

The previous work \cite{Victor2015b} optimized the pilot powers to maximize functions of the SE, but the paper only considers a single cell with orthogonal pilot sequences, i.e., $\tau_p \geq KL$ with $L=1$. Besides, the authors of \cite{Guo2015a} optimized the pilot powers to maximize the energy efficiency of a multi-cell system. That paper  assumed $\tau_p = K$ and a fixed pilot assignment. If $p_{l,k}$ is the pilot power of user~$k$ in cell~$l$, then the square root of the power matrix allocated to the $K$~users in cell~$l$ is a diagonal matrix defined as
 \begin{equation}
\widetilde{\pmb{P}}_l = \mathrm{diag}\left( \sqrt{\hat{p}_{l,1}}, \ldots,  \sqrt{\hat{p}_{l,K}} \right).
 \end{equation}
The pilot used for the users in cell~$l$ is then formulated as
\begin{equation} \label{eq: PilotStructure2}
\widetilde{\mathbf{\Psi}_l} = \pmb{\Phi} \widetilde{\pmb{P}}_l.
\end{equation}
Similar to \eqref{eq:Max-Power}, the pilot power at user~$k$ in cell~$l$ is limited as
\begin{equation}
 0 \leq \hat{p}_{l,k} \leq \tau_p P_{\max,l,k}.
\end{equation}
This is also a special case of our proposed pilot design, since \eqref{eq: PilotStructure2} assumes orthogonal pilots and fixed pilot assignment.

If we combine the pilot structure in \eqref{eq: PilotStructure2} with the permutation matrix approach from \eqref{eq: fixedPilotPower}, the pilot sequences of the $K$~users in cell~$l$ become
\begin{equation} \label{eq: PilotStructure3}
\widetilde{\widetilde{\mathbf{\Psi}}}_l =  [\tilde{\tilde{\pmb{\psi}}}_{l,1}, \ldots, \tilde{\tilde{\pmb{\psi}}}_{l,K} ] = \pmb{\Phi} \pmb{\Pi}_l \widetilde{\pmb{P}}_l .
\end{equation}
In principle, we can now consider all possible permutation matrices and optimize the pilot power for each one, based on the algorithms in previous work. This approach is computationally heavy, but will serve as a benchmark in Section~\ref{Section: Experimental Result}.

\section{UL Massive MIMO With Arbitrary Pilots} \label{Section: ULTransmission}
This section provides ergodic SE expressions with the new pilot sequences in \eqref{eq: PilotStructure1}, which will be used for optimized pilot design and power control in Section~\ref{Section: OptProblem}.
\setcounter{eqnback}{\value{equation}} \setcounter{equation}{22}
 \begin{figure*}[t]
 	%\begin{small}
 	\begin{equation} \label{eq: SINR_MRC1}
 	\mathrm{SINR}_{l,k}= \frac{ M (\beta_{l,k}^l)^2 p_{l,k} \left( \sum\limits_{b=1}^{\tau_p} \hat{p}_{l,k}^b \right)^2 }{  \left( \sum\limits_{(i,t) \in \mathcal{S} } \beta_{i,t}^l \left( \sum\limits_{b=1}^{\tau_p} \sqrt{ \hat{p}_{i,t}^{b}  \hat{p}_{l,k}^{b}} \right)^2  +  \sigma^2 \sum\limits_{b=1}^{\tau_p} \hat{p}_{l,k}^b \right) \left( \sum\limits_{(i,t) \in \mathcal{S} } p_{i,t}  \beta_{i,t}^l + \sigma^2 \right)   +  M \sum\limits_{(i,t) \in \mathcal{S} \setminus (l,k)} p_{i,t}  (\beta_{i,t}^l )^2  \left(\sum\limits_{b=1}^{\tau_p} \sqrt{ \hat{p}_{i,t}^{b}  \hat{p}_{l,k}^{b}} \right)^2  }
 	\end{equation}
 	\vspace*{-0.25cm}
 	\hrulefill
 	\vspace*{-0.25cm}
 \end{figure*}
 \setcounter{eqncnt}{\value{equation}}
 \setcounter{equation}{\value{eqnback}}
\subsection{Channel Estimation}
During the UL pilot transmission, the received signal $\mathbf{Y}_l \in \mathbb{C}^{M \times \tau_p}$ at the BS of cell~$l$ is 
\begin{equation} \label{eq: Received_Pilot}
\mathbf{Y}_l  = \sum_{(i,t) \in \mathcal{S}} \mathbf{h}_{i,t}^{l} \pmb{\psi}_{i,t}^{H} +  \mathbf{N}_l,
\end{equation}
where $\mathbf{h}_{i,t}^l \in \mathbb{C}^M$ denotes the channel between user~$t$ in cell~$i$ and BS~$l$. $\mathbf{N}_l \in \mathbb{C}^{M \times \tau_p}$ is the additive noise with independent elements distributed as $\mathcal{CN}(0, \sigma^2)$. Correlating $\mathbf{Y}_l$ in \eqref{eq: Received_Pilot} with pilot sequence $\pmb{\psi}_{l,k}$ of user~$k$ in cell~$l$, we obtain
\begin{equation}
\mathbf{y}_{l,k} = \mathbf{Y}_{l,k} \pmb{\psi}_{l,k} = \sum_{(i,t) \in \mathcal{S}} \mathbf{h}_{i,t}^{l} \pmb{\psi}_{i,t}^H \pmb{\psi}_{l,k} + \mathbf{N}_{l} \pmb{\psi}_{l,k}.
\end{equation}
We consider independent Rayleigh fading where the channel between user~$t$ in cell~$i$ and BS~$l$ is distributed as
\begin{equation}
\mathbf{h}_{i,t}^l \sim \mathcal{CN} \left(  \mathbf{0} ,  \beta_{i,t}^l \mathbf{I}_M \right),
\end{equation}
where the variance $\beta_{i,t}^l$ determines the large-scale fading, including geometric attenuation and shadowing. By using minimum mean squared error (MMSE) estimation, the distributions of the channel estimate and estimation error are as follows.

\begin{lemma} \label{lemma: Distribution}
If the system uses the pilot structure in \eqref{eq: PilotStructure1}, the channel estimate is distributed as
\begin{equation}
\hat{ \mathbf{h} }_{l,k}^l \sim \mathcal{CN} \left( \mathbf{0}, \gamma_{l,k}^l   \mathbf{I}_M  \right),
\end{equation}
where
\begin{equation*}
\gamma_{l,k}^l = \frac{ (\beta_{l,k}^l)^2 \left(\sum\limits_{b=1}^{\tau_p} \hat{p}_{l,k}^b \right)^2 }{ \sum\limits_{(i,t) \in \mathcal{S} } \beta_{i,t}^l \left( \sum\limits_{b=1}^{\tau_p} \sqrt{\hat{p}_{i,t}^b  \hat{p}_{l,k}^b} \right)^2  +  \sigma^2 \sum\limits_{b=1}^{\tau_p} \hat{p}_{l,k}^b }.
\end{equation*}
The estimation error $\mathbf{e}_{l,k}^l  = \mathbf{h}_{l,k}^l - \hat{\mathbf{h}}_{l,k}^l$ is independent of the channel estimate and distributed as
\begin{equation}
\mathbf{e}_{l,k}^l
\sim \mathcal{CN} \left( \mathbf{0} , \left( \beta_{l,k}^l - \gamma_{l,k}^l \right)\mathbf{I}_M  \right).
\end{equation}
\end{lemma}
\begin{proof}
This result follows directly from standard MMSE estimation in \cite{Kay1993a}.
\end{proof}

Lemma~\ref{lemma: Distribution} provides the MMSE estimator for the general pilot structure in \eqref{eq: PilotStructure1}. The pilot powers as well as the inner products between pilot sequences appear explicitly in the expressions.

\subsection{UL Data Transmission}
In the UL data transmission, user~$t$ in cell~$i$ transmits the signal $x_{i,t} \sim \mathcal{CN}(0,1)$. The $M \times 1$ received signal vector at BS~$l$ is the superposition of the transmitted signals
\begin{equation}
\mathbf{y}_l = \sum_{(i,t) \in \mathcal{S} } \sqrt{p_{i,t}} \mathbf{h}_{i,t}^l x_{i,t} + \mathbf{n}_l,
\end{equation}
where $p_{i,t}$ is the transmit power corresponding to the signal $x_{i,t}$ and the additive noise is $\mathbf{n}_l \sim \mathcal{CN} ( \mathbf{0}, \sigma^2 \mathbf{I}_M)$. To detect the transmitted signal, BS~$l$ selects a detection vector $\mathbf{v}_{l,k} \in \mathbb{C}^M$ and applies it to the received signal as
\begin{equation} \label{eq: Signal-Detection}
\mathbf{v}_{l,k}^H \mathbf{y}_l = \sum_{(i,t) \in \mathcal{S} } \sqrt{p_{i,t}}  \mathbf{v}_{l,k}^H \mathbf{h}_{i,t}^l x_{i,t} +  \mathbf{v}_{l,k}^H \mathbf{n}_l .
\end{equation}
A general lower bound on the UL ergodic capacity of user~$k$ in cell~$l$ is computed in \cite{Bjornson2016a} as
\begin{equation} \label{eq:RateProposedPilot}
R_{l,k} = \left( 1 - \frac{\tau_p}{\tau_c} \right) \log_2 \left(1 + \textrm{SINR}_{l,k} \right),
\end{equation}
with $\textrm{SINR}_{l,k}$ given by
\begin{equation} \label{eq: SINR_k}
  \frac{ p_{l,k} | \mathbb{E} \{ \mathbf{v}_{l,k}^{H} \mathbf{h}_{l,k}^l  \} |^2 }{\sum\limits_{(i,t) \in \mathcal{S} } p_{i,t} \mathbb{E} \{ | \mathbf{v}_{l,k}^{H} \mathbf{h}_{i,t}^{l} |^2 \} - p_{l,k} | \mathbb{E} \{ \mathbf{v}_{l,k}^{H} \mathbf{h}_{l,k}^{l} \} |^2 + \sigma^2 \mathbb{E} \{ \| \mathbf{v}_{l,k} \|^2 \} }.
\end{equation}

As a contribution of this paper, we compute a closed-form expression for this lower bound in the case of MR detection with $\mathbf{v}_{l,k} = \hat{\mathbf{h}}_{l,k}^l$. This is a highly scalable detection method suitable for practical Massive MIMO systems.

\begin{lemma} \label{Lemma: Achievable_Rate}
If the system uses the pilot structure in \eqref{eq: PilotStructure1} and MR detection, the SE in \eqref{eq:RateProposedPilot} for user~$k$ in cell~$l$ becomes
	\begin{equation}
	R_{l,k} =  \left( 1 - \frac{\tau_p}{\tau_c} \right) \log_2 \left(1 + \mathrm{SINR}_{l,k} \right),
	\end{equation}
	where  $\mathrm{SINR}_{l,k}$ is shown in \eqref{eq: SINR_MRC1} at the top of this page.
\end{lemma}
\begin{proof}
The SINR value in \eqref{eq: SINR_MRC1} is obtained by computing the moments of Gaussian distributions, similar to \cite{Chien2017a}. The detailed proof is omitted due to space limitations. 
\end{proof}

Inspecting \eqref{eq: SINR_MRC1}, we notice that it is always advantageous to add BS antennas since the numerator grows linearly with $M$. The first term in the denominator represents non-coherent interference from all users in the system,  and it is independent of $M$. The second term in the denominator represents coherent interference caused by pilot contamination and it grows linearly with $M$. We stress that a proper pilot design and power control $\hat{p}_{l,k}^b, \forall l,k,b,$ can improve the SE by enhancing the channel estimation quality and reducing the coherent interference caused by pilot contamination.

\section{Max-min Fairness Optimization} \label{Section: OptProblem}
In this section, we utilize the SE expression in Lemma~\ref{Lemma: Achievable_Rate} to formulate a max-min SE pilot optimization problem. We further demonstrate that the optimization problem is NP-hard, and therefore instead of seeking the global optimum, a local solution with polynomial complexity is derived.

\setcounter{eqnback}{\value{equation}} \setcounter{equation}{30}
 \begin{figure*}[t]
 	\begin{equation} \label{eq: SINR_MRCApproximation}
 	\widetilde{\textrm{SINR}}_{l,k}= \frac{ M (\beta_{l,k}^l)^2 p_{l,k} \prod\limits_{b=1}^{\tau_p} \left( \hat{p}_{l,k}^b / \alpha_{l,k}^b \right)^{2\alpha_{l,k}^b} }{  \left( \sum\limits_{(i,t) \in \mathcal{S} } \beta_{i,t}^l \left( \sum\limits_{b=1}^{\tau_p} \sqrt{ \hat{p}_{i,t}^{b}  \hat{p}_{l,k}^{b}} \right)^2  +  \sigma^2 \sum\limits_{b=1}^{\tau_p} \hat{p}_{l,k}^b \right) \left( \sum\limits_{(i,t) \in \mathcal{S} } p_{i,t}  \beta_{i,t}^l + \sigma^2 \right)   +  M \sum\limits_{(i,t) \in \mathcal{S} \setminus (l,k)} p_{i,t}  (\beta_{i,t}^l )^2  \left(\sum\limits_{b=1}^{\tau_p} \sqrt{ \hat{p}_{i,t}^{b}  \hat{p}_{l,k}^{b}} \right)^2  }
 	\end{equation}
 	\vspace*{-0.25cm}
 	\hrulefill
 	\vspace*{-0.25cm}
 \end{figure*}
  \setcounter{eqncnt}{\value{equation}}
  \setcounter{equation}{\value{eqnback}}
  
\subsection{Problem Formulation} \label{Subsect: Problem}
One of the key visions of Massive MIMO is to provide uniformly good service for everyone in the system, which is known as max-min fairness. In this paper, we investigate how to optimize the pilot sequences to achieve this goal. We consider the pilot powers (over the basis vectors) as optimization variables while the data powers are assumed to be predetermined. The max-min SE optimization problem is formulated for the proposed pilot design as
\setcounter{eqnback}{\value{equation}} \setcounter{equation}{23}
\begin{equation} \label{eq: Opt_Prob1}
\begin{aligned}
&\underset{\{ \hat{p}_{l,k}^b \geq 0 \}}{ \mathrm{maximize} } &&  \underset{(l,k)}{\min} \;  \log_2 \left( 1 + \textrm{SINR}_{l,k} \right) \\
& \text{subject to} && \frac{1}{\tau_p}\sum_{b=1}^{\tau_p} \hat{p}_{l,k}^b \leq P_{\max, l,k}, \forall l,k.
\end{aligned}
\end{equation}
Note that this optimization problem jointly generates the pilot sequences and performs pilot power control. An equivalent epigraph-form representation of \eqref{eq: Opt_Prob1} is
\begin{subequations} \label{eq: Opt_Prob2}
\begin{align} 
& \underset{ \xi, \{ \hat{p}_{l,k}^b \geq 0 \}}{ \mathrm{maximize} }  && \xi \\
& \text{subject to} &&  \mathrm{SINR}_{l,k}  \geq \xi, \forall l,k, \label{P1:a} \\
&&& \frac{1}{\tau_p}\sum_{b=1}^{\tau_p} \hat{p}_{l,k}^b \leq P_{\max, l,k}, \forall l,k. \label{P1:b} 
\end{align}
\end{subequations}
From the expression of the SINR constraints in \eqref{eq: Opt_Prob2}, we realize that the proposed max-min SE optimization problem is a signomial program.\footnote{A function $f(x_1, \ldots, x_{N_1}) = \sum_{n=1}^{N_2} c_n \prod_{m= 1}^{N_1} x_m^{a_{n,m}}$ defined in $\mathbb{R}_{+}^{N_1}$ is signomial with $N_2$ terms $(N_2 \geq 2)$ if the exponents $a_{n,m}$ are real numbers and the coefficients $c_n$ are also real but at least one of them must be negative. In case of all $c_n, \forall n,$ are positive, $f(x_1, \ldots, x_{N_1})$ is a posynomial function.} Therefore, the max-min SE optimization problem is NP-hard in general and seeking the optimal solution has very high complexity in any non-trivial setup \cite{Lange2014a}. However, the power constraints \eqref{P1:b} ensure a compact feasible domain and make the SINRs continuous functions so that the optimal solution to \eqref{eq: Opt_Prob2} always exists.

\subsection{Local Optimality Algorithm}
This subsection provides an algorithm to approximate the optimization problem \eqref{eq: Opt_Prob2} as a geometric program. In detail, the signomial SINR constraints are converted to corresponding monomial constraints by using the weighted arithmetic mean-geometric mean inequality \cite{Chiang2007b} as in Lemma~\ref{Lemma: Local_Approximation}.\footnote{ A function $f(x_1, \ldots, x_{N_1}) = c\prod_{m=1}^{N_1} x_m^{a_m}$ defined in $\mathbb{R}_{+}^{N_1}$ is monomial if the coefficient $c >0$ and the exponients $a_m, \forall m,$ are real numbers. }
\begin{lemma} \cite[Lemma~1]{Chiang2007b}  \label{Lemma: Local_Approximation}
Assume that a posynomial function $g(x)$ is defined from the set of $\tau_p$ monomials  $\{ u_1 (x), \ldots, u_{\tau_p} (x) \}$
\begin{equation}
g(x) = \sum_{b=1}^{\tau_p} u_b (x),
\end{equation}
then this posynomial function is lower bounded by a monomial function $\tilde{g}(x)$ as
\begin{equation}
g(x) \geq \tilde{g}(x) = \prod_{b=1}^{\tau_p} \left(  u_{b}(x)/ \alpha_b \right)^{\alpha_b},
\end{equation}
where $\alpha_b$ is a non-negative weight value corresponding to $u_{b} (x)$. We say that $\tilde{g}(x_0)$ is the best approximation to $g(x_0)$ near the given point $x_0$ in the sense of the first order Taylor expansion, if the weight $\alpha_b $  is defined as
\begin{equation} \label{eq: WeightDef}
\alpha_b = u_b(x_0) \Big/ \sum_{b=1}^{\tau_p} u_b (x_0) .
\end{equation}
\end{lemma}
Using this lemma, the max-min SE optimization problem \eqref{eq: Opt_Prob2} is converted to a geometric program by bounding the term $\sum_{b=1}^{\tau_p}  \hat{p}_{l,k}^b$ in the numerators of the SINR constraints:
\begin{equation} \label{eq_: Power_Approximation}
 \sum_{b=1}^{\tau_p}  \hat{p}_{l,k}^b \geq \prod_{b=1}^{\tau_p} \left( \hat{p}_{l,k}^b / \alpha_{l,k}^b \right)^{\alpha_{l,k}^b},
\end{equation}
where $\alpha_{l,k}^b$ is the weight value corresponding to $\hat{p}_{l,k}^b$. It leads to a lower bound on the SINR value for user $k$ in cell $l$ as
\begin{equation} \label{eq: SINRBound}
 \textrm{SINR}_{l,k} \geq   \widetilde{\textrm{SINR}}_{l,k},
\end{equation}
where $\widetilde{\textrm{SINR}}_{l,k}$ is presented in \eqref{eq: SINR_MRCApproximation} at the top of this page.

The solution to the max-min SE optimization problem \eqref{eq: Opt_Prob2} is lower bounded by the following geometric program
\setcounter{eqnback}{\value{equation}} \setcounter{equation}{31}
\begin{equation} \label{eq: Opt_Prob3}
\begin{aligned}
& \underset{ \xi, \{ \hat{p}_{l,k}^b \geq 0 \}}{ \mathrm{maximize} }  && \xi \\
& \text{subject to} &&  \widetilde{\mathrm{SINR}}_{l,k}  \geq \xi, \forall l,k, \\
&&& \frac{1}{\tau_p}\sum_{b=1}^{\tau_p} \hat{p}_{l,k}^b \leq P_{\max, l,k}, \forall l,k.
\end{aligned}
\end{equation}
By virtue of the successive approximation technique \cite{Marques1978a}, a local solution to the original optimization problem \eqref{eq: Opt_Prob2} is obtained  if we solve \eqref{eq: Opt_Prob3} iteratively as follows.

\begin{theorem} \label{Theorem: KKTpoint}
Selecting a feasible starting point $\hat{p}_{l,k}^{b, (0)}, \forall l,k,b,$ and solving \eqref{eq: Opt_Prob3} in an iterative manner via consecutively updating the weight values $\alpha_{l,k}^b$, the solution will converge to the Karush-Kuhn-Tucker (KKT) local point to \eqref{eq: Opt_Prob2}.
\end{theorem}
\begin{proof}
The proof is adapted from the general framework in \cite{Marques1978a}. We first prove that the procedures in Theorem \ref{Theorem: KKTpoint} guarantee that the solution converges to a limit point. This point is further proved to be a KKT local point to \eqref{eq: Opt_Prob2}. The detail proof is omitted due to space limitations.
\end{proof}

After selecting the initial powers $\hat{p}_{l,k}^{b,(0)}, \forall l,k,b$, we compute the weight values by applying \eqref{eq: WeightDef}. Furthermore, in each iteration, the SINR constraints are converted to the monomials by bounding the pilot power of user~$k$ in cell~$l$ as in \eqref{eq: SINR_MRCApproximation} with noting that the weight values are computed based on the optimal powers of the previous iteration using \eqref{eq: WeightDef}. The solution is then obtained by solving the geometric program \eqref{eq: Opt_Prob3}. At the end of each iteration, the weight values are updated for the next iteration. We repeat the procedure until the algorithm converged to a KKT local point. The convergence can be declared, for example, when the variation between two consecutive iterations is sufficient small. The proposed local optimality approach is summarized  in Algorithm \ref{Algorithm: Local_Approximation}. 
\begin{algorithm}
\caption{Successive approximation algorithm for \eqref{eq: Opt_Prob2}} \label{Algorithm: Local_Approximation}
\textbf{Input}: Set $i=1$; Select the data powers $p_{l,k}$ for $ \forall l= 1, \ldots, L; k =1, \ldots, K$; Select the initial values of powers $\hat{p}_{l,k}^{b,(0)}$ for $\forall l=1, \ldots, L; k= 1, \ldots, K,$ and $b=1,\ldots,\tau_p$; Compute the weight values: $\alpha_{l,k}^{b,(1)} = \hat{p}_{l,k}^{b,(0)} /\sum_{b=1}^{\tau_p} \hat{p}_{l,k}^{b,(0)}, \forall l,k,b.$
\begin{itemize}
  \item[1.] \emph{Iteration} $i$: 
  \begin{itemize}
  
    \item[1.1.] Solve the geometric program \eqref{eq: Opt_Prob3} with $\alpha_{l,k}^b = \alpha_{l,k}^{b,(i)}$ to get the optimal values $\xi^{(i),\ast}$ and $\hat{p}_{l,k}^{b,(i),\ast}, \forall l,k,b$.

    \item[1.2.] Update the weight values: $ \alpha_{l,k}^{b,(i+1)} = \hat{p}_{l,k}^{b,(i),\ast} /\sum_{b=1}^{\tau_p} \hat{p}_{l,k}^{b,(i),\ast}, \forall l,k,b.$
 \end{itemize}
 \item[2.] If Stopping criterion satisfied  $\rightarrow$ Stop. Otherwise, go to Step 3.
   \item[3.] Set $\xi^{\ast} = \xi^{(i),\ast}$ and $\hat{p}_{l,k}^{b,\ast} = \hat{p}_{l,k}^{b,(i),\ast}, \forall l,k,b$; Set $i = i+1$, go to Step 1.
\end{itemize}
\textbf{Output}: The solutions $\xi^{\ast}$ and $\hat{p}_{l,k}^{b,\ast}, \forall l,k,b.$
\end{algorithm}

%=====================================================
%=====================================================
\section{Experimental Results} \label{Section: Experimental Result}
\vspace{-0.1cm}
A Massive MIMO system with coverage area $1 \mbox{ km}^2$ comprising of $4$ square cells is considered for simulation. In each cell, a BS is located at the center, while  $K$ users are uniformly distributed at distance greater than $35$ m from the BS. To even out the interference, the coverage area is wrapped around, and therefore one BS has eight neighbors. 
We assume that the coherence block contains $\tau_c =200$ symbols. The system operates over a $20$ MHz bandwidth and the corresponding noise variance is $-96$ dBm, including a noise figure of $5$~dB. The large-scale fading coefficient $\beta_{i,t}^l $ is computed as $\beta_{i,t}^l = -148.1 - 37.6 \log_{10} d_{i,t}^l + z_{i,t}^l$ [dB], where $d_{i,t}^l$ denotes the distance [km] between user $t$ in cell $i$ and BS $l$. The shadow fading $z_{i,t}^l$ is created by a Gaussian distributed with zero mean and standard derivation $7$ dB.\footnote{ Shadow fading realizations were sometimes regenerated to ensure that the home BS has the largest large-scale fading to its users (i.e., $\beta_{l,k}^l$ is the maximum over all $\beta_{i,k}^l, i =1,\ldots, L$.) } The payload data symbols have equal power, $p_{l,k}=200$ mW, $\forall l,k$ and the maximum pilot power constraints $P_{\max,l,k} = 200$ mW, $\forall l,k$. 

For Algorithm \ref{Algorithm: Local_Approximation}, we observed better performance with a hierarchical initialization of $\hat{p}_{l,k}^{b,(0)}$ than with an all-equal initialization. Consequently, we initialize $\hat{p}_{l,k}^{b,(0)}$ as uniformly distributed over the range $[0, P_{\max,l,k}]$. Algorithm \ref{Algorithm: Local_Approximation} converges quite fast, so the stopping criteria can be easily specified in the number of iterations (e.g., $15$ iterations). The proposed algorithm is compared with related works and brute-force:
\begin{itemize}
\item[$(i)$] \emph{Universal random pilot assignment}, as considered in \cite{Jose2011b,Bjornson2016a}. The same pilots are reused in every cell and assigned randomly to the users within the cell. Equal pilot power  $\hat{p} = 200$ mW is used by all users. 
 \item[$(ii)$] \emph{Smart pilot assignment}, as proposed in \cite{Xu2015a}. Orthogonal pilot sequences are assumed in every cell and are assigned to the users based on the mutual interference, determined by the large-scale fading coefficients. Equal pilot power  $\hat{p} = 200$ mW is used by all users. 
 \item[$(iii)$] \emph{Pilot power control with brute-force search} utilizes the pilot structure in \eqref{eq: PilotStructure3}. A brute-force search over all permutation matrices  $\pmb{\Pi}_l$ is performed, and for each matrix the optimum pilot powers are computed.
 \end{itemize}
 \begin{figure}[t]
 	\centering
	\includegraphics[ trim=0.5cm 0cm 1.2cm 0.55cm, clip=true, width=3.2in]{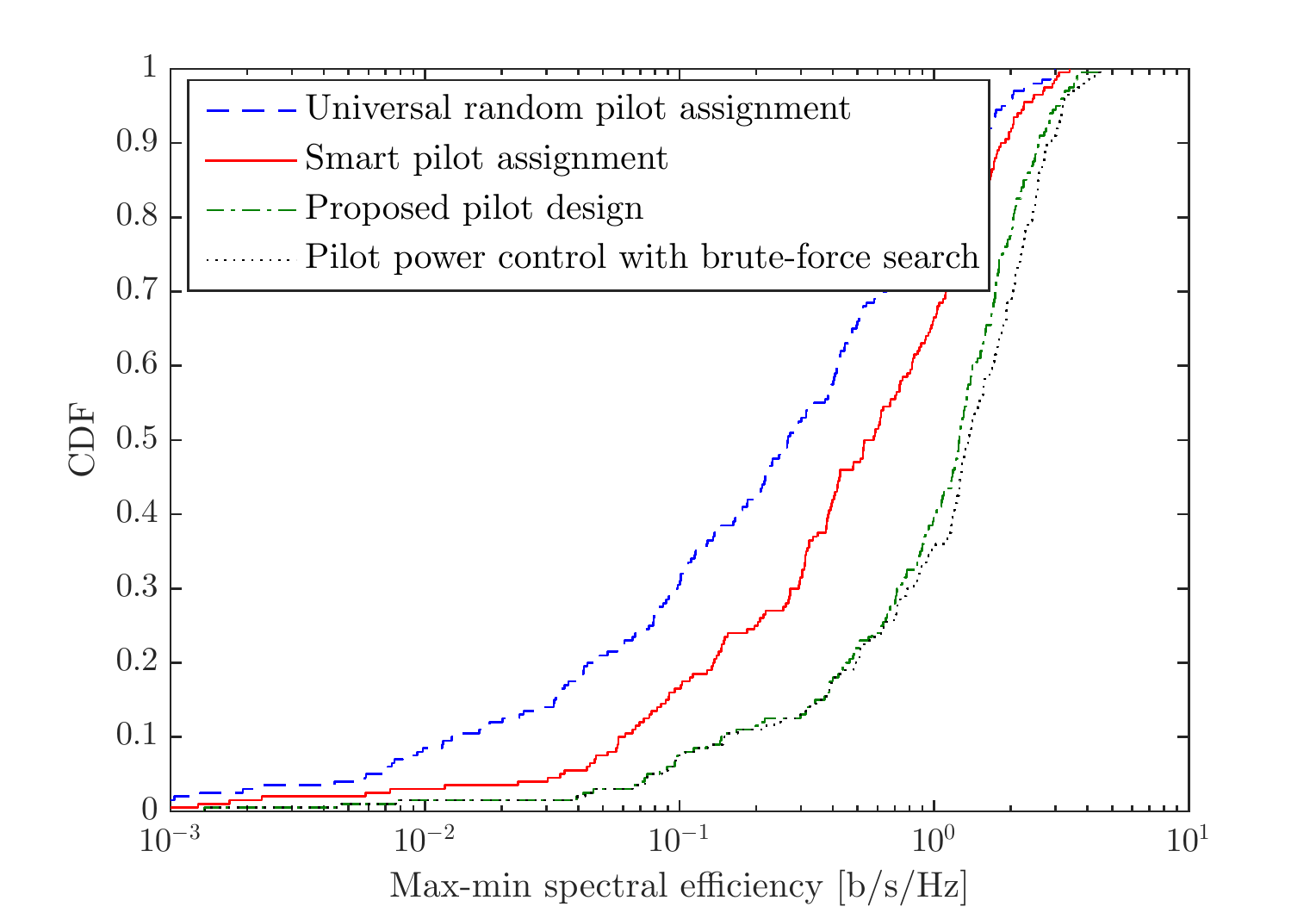} \vspace*{-0.4cm}
 	\caption{ Cumulative distribution function (CDF) of the max-min SE [b/s/Hz] with $K = \tau_p = 2$ and $M = 300$.}
 	\label{Fig-CDF-2K2B}
 	\vspace*{-0.45cm}
 \end{figure}
SE is measured over different random user locations and shadow fading realizations. The SE achieved by $(i)-(iii)$ and Algorithm \ref{Algorithm: Local_Approximation} are also averaged over different pilot reuse locations and initializations of $\hat{p}_{l,k}^{b,(0)}, \forall l,k,b$, respectively. Additionally, the solutions to the optimization problems are obtained by utilizing the MOSEK solver \cite{Mosek} with CVX \cite{cvx2015}.

Fig.~\ref{Fig-CDF-2K2B} shows the cumulative distribution function (CDF) of the max-min SE [b/s/Hz] for the case $K=\tau_p=2$ and $M =300$. Universal random pilot assignment yields the worst SE, because of the pilot contamination and mutual interference. At the $95 \%$-likely SE point, smart pilot assignment brings significant enhancement: it is about $4.75\times$ better than universal random pilot assignment thanks to exploitation of the mutual interference between the users \cite{Xu2015a}. Although the performance of smart pilot assignment is very close to optimal pilot assignment with brute-force search for a fixed power level \cite{Xu2015a}, by jointly optimizing the power and pilot assignment, the proposed method outperforms smart pilot assignment by providing a $1.6\times$ gain in average max-min SE. Furthermore, the similar performance of the proposed pilot design and pilot power control with brute-force search confirms effectiveness of the proposed local optimality algorithm. 

\begin{figure}[t]
	\centering
	\includegraphics[trim=0.5cm 0cm 1.2cm 0.55cm, clip=true, width=3.2in]{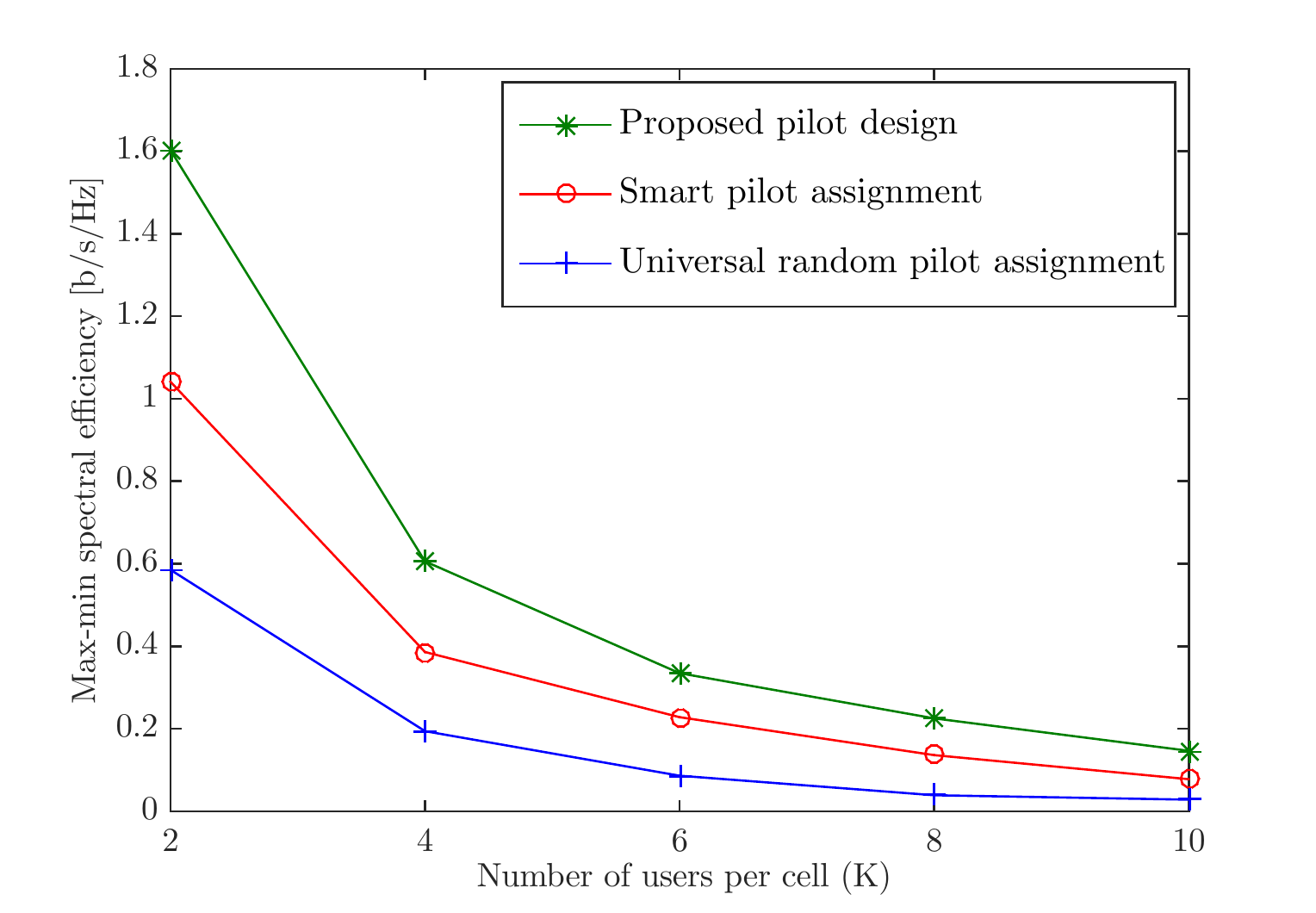} \vspace*{-0.4cm}
	\caption{ Max-min SE [b/s/Hz] vs. the number of user per cell with $K = \tau_p$ and $M = 300$.}
	\label{Fig-VariousK}
	\vspace*{-0.4cm}
\end{figure}

\begin{figure}[t]
	\centering
	\includegraphics[trim=0.5cm 0cm 1.2cm 0.55cm, clip=true, width=3.2in]{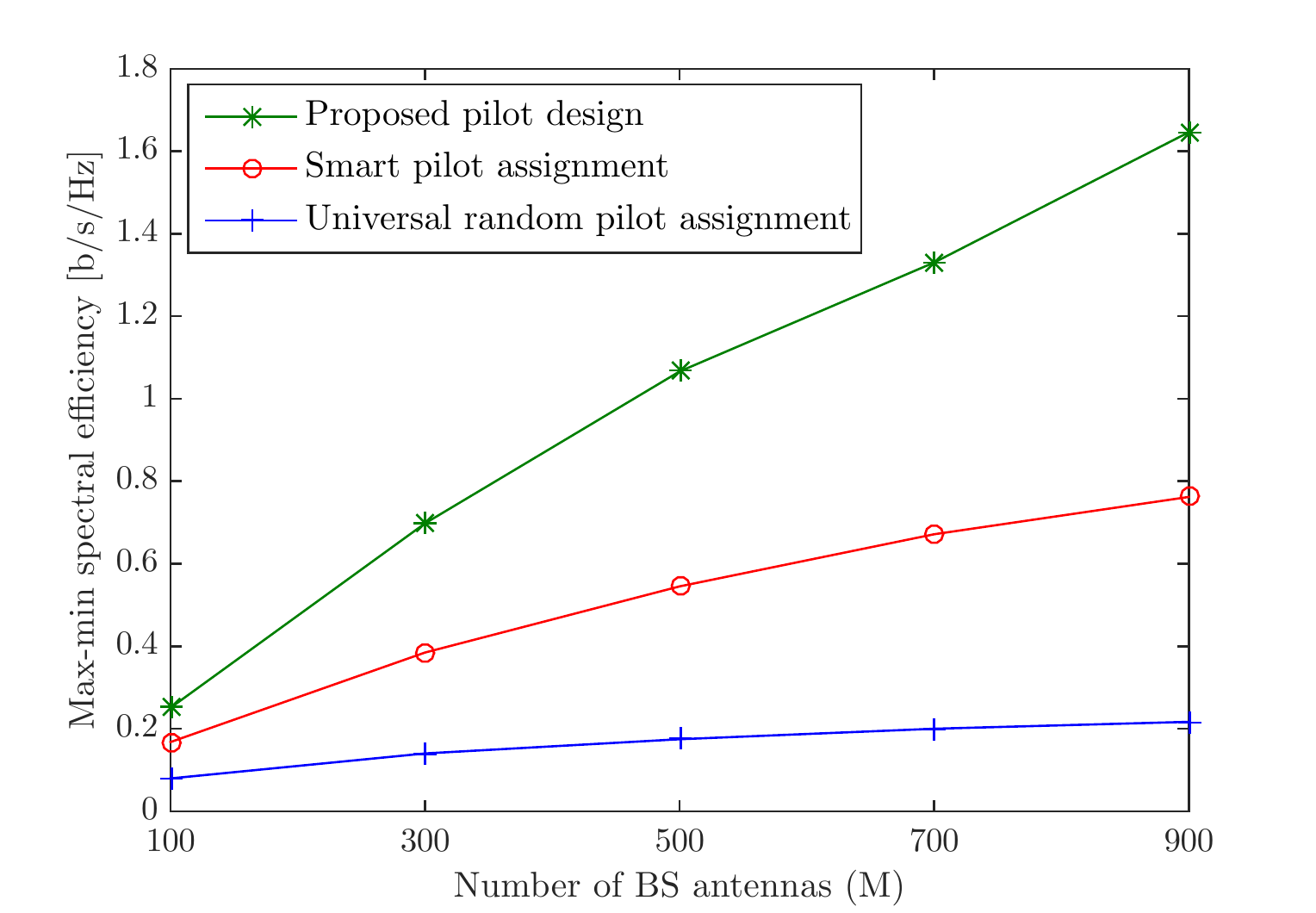} \vspace*{-0.4cm}
	\caption{ Max-min SE [b/s/Hz] vs. the number of BS antennas, $ K = \tau_p = 4$.}
	\label{Fig-VariousM}
	 \vspace*{-0.45cm}
\end{figure}

Due to huge computational complexity, the brute-force search is not considered hereafter when we increase the number of users. Fig.~\ref{Fig-VariousK} plots the average max-min SE as a function of the number of users per cell, assuming $\tau_p=K$. The proposed pilot design provides the highest SE over all tested scenarios. Specifically, in comparison to universal random pilot assignment, the improvement varies from $2.73\times$  to $5.22\times$  with $K = 2 $ to $K=10$, respectively. Even though smart pilot assignment performs better than universal random pilot assignment, the proposed method still provides SE improvements of up to $1.88\times $ at $K = 10$. Moreover, we observe a dramatic reduction of the max-min SE when the number of users increases due to stronger mutual interference.

Fig.~\ref{Fig-VariousM} shows the average max-min SE versus the number of BS antennas. Among the three pilot assignment techniques, we again observe the worst SE with universal random pilot reuse. The max-min SE increases from $0.08$ [b/s/Hz] to $0.22$ [b/s/Hz] from $M=100$ to $M=900$. Our proposed pilot design always yields the highest SE and the gap to the smart pilot assignment reaches up to $2.16 \times$ at $M = 900$. 

\section{Conclusion} \label{Section: Conclusion}
\vspace*{-0.1cm}
This paper proposed a new methodology for joint optimization of the pilot assignment and pilot power control in Massive MIMO systems. The key difference from prior work is to treat the pilot sequences as continuous optimization variables, instead of predefined vectors that should be assigned combinatorially. A new SE expression was computed for the proposed pilot structure and it was used to formulate a max-min SE optimization problem. Finding the globally optimal solution is NP-hard, but we obtained an efficient local optimum that outperforms the previous state-of-the-art methods for pilot assignment. Large gains in max-min SE can be achieved by the proposed pilot assignment.

%==========================Reference==========================================================================
\bibliographystyle{IEEEtran}
\bibliography{IEEEabrv,refs}
\end{document}